\documentclass[letterpaper, 10 pt, conference]{ieeeconf}  

\IEEEoverridecommandlockouts
\usepackage{cite}
\usepackage{amsmath,amssymb,amsfonts}
\usepackage{algorithmic}
\usepackage{graphicx}
\usepackage{algorithm,algorithmic}
\usepackage{nicematrix}
 \usepackage{flushend}
\usepackage{mathtools}
\usepackage{hyperref}
\usepackage{subcaption}
\hypersetup{
    colorlinks=true,
    linkcolor=blue,
    filecolor=magenta,      
    urlcolor=cyan,
    pdftitle={Overleaf Example},
    pdfpagemode=FullScreen,
    }
\usepackage{textcomp}

\newcommand{\ie}{\textit{i.e.,~}}

\newcommand{\ack}{\texttt{ACK}}
\newcommand{\nack}{\texttt{NACK}}

\newtheorem{theorem}{Theorem}
\newtheorem{remark}{Remark}

\usepackage{nicematrix}
\usepackage[htt]{hyphenat}

\makeatletter
\newcommand*\bigcdot{\mathpalette\bigcdot@{.8}}
\newcommand*\bigcdot@[2]{\mathbin{\vcenter{\hbox{\scalebox{#2}{$\m@th#1\bullet$}}}}}
\makeatother

\newenvironment{list4}{
	\begin{list}{$\bullet$}{%
			\setlength{\itemsep}{0.05cm}
			\setlength{\labelsep}{0.2cm}
			\setlength{\labelwidth}{0.3cm}
			\setlength{\parsep}{0in} 
			\setlength{\parskip}{0in}
			\setlength{\topsep}{0in} 
			\setlength{\partopsep}{0in}
			\setlength{\leftmargin}{0.16in}}}
	{\end{list}}

\begin{document}
\title{Remote Estimation over Packet-Dropping Wireless Channels\\ with Partial State Information}

\author{Ioannis~Tzortzis, Evagoras~Makridis, Charalambos~D.~Charalambous, and Themistoklis~Charalambous  
\thanks{The authors are with the Department of Electrical and Computer Engineering, School of Engineering, University of Cyprus, 1678 Nicosia, Cyprus.  E-mails: \texttt{\{tzortzis.ioannis,~makridis.evagoras, chadcha,~charalambous.themistoklis\}@ucy.ac.cy}.\newline T. Charalambous is also a Visiting Professor at the Department of Electrical Engineering and Automation, School of Electrical Engineering, Aalto University, 02150 Espoo, Finland.}
\thanks{This work has been partly funded by MINERVA, a European Research Council (ERC) project funded under the European Union's Horizon 2022 research and innovation programme (Grant agreement No. 101044629).}
}

\maketitle

\begin{abstract}
In this paper, we study the design of an optimal transmission policy for remote state estimation over packet-dropping wireless channels with imperfect channel state information. A smart sensor uses a Kalman filter to estimate  the system state and transmits its information to a remote estimator. Our objective is to minimize the state estimation error and energy consumption by deciding whether to transmit new information or retransmit previously failed packets. To balance the trade-off between information freshness and reliability, the sensor applies a hybrid automatic repeat request  protocol. We formulate this problem as a  finite horizon partially observable Markov decision process  with an augmented state-space that incorporates both the age of information  and the unknown channel state. By defining an information state,  we derive the dynamic programming equations for evaluating the optimal policy. This transmission policy is computed numerically using the point-based value iteration algorithm.
\end{abstract}


\section{Introduction}
\label{sec:introduction}
The penetration of smart sensor technology into mission-critical applications such as industrial automation, smart grid monitoring and control, remote healthcare, and others, enabled the transition from on-device monitoring to remote estimation and monitoring. Unlike traditional sensors that merely collect raw data, smart sensors combine sensing with integrated computational and communication capabilities. This allows smart sensors to additionally process the data locally and transmit only useful information to a remote monitor or decision-maker \cite{cocco2023remote,munari2023remote,li2019age,Tahmoores:2023TAC}. In mission-critical applications, smart sensors can effectively reduce the amount of data that need to be sent over the network towards the destined receiver, with the aim of reliably communicating pre-processed information in a timely manner. 

In networked control systems, smart sensors often employ data fusion and state estimation algorithms like Kalman filters to pre-estimate the states of a dynamical system \cite{liu2021remote,chakravorty2019remote,wang2021remote,qu2023remote}, and communicate them to the remote estimator over the network. This approach offers superior state estimation accuracy compared to traditional sensors that merely transmit raw measurement data, especially under unreliable and constrained communication channels \cite{qu2023remote}. However, communicating critical information over unreliable wireless channels can lead to severe communication impediments such as delays or even loss of information, which subsequently affect the performance of the remote estimator. These impediments arise due to channel fading, noise, interference and path loss, affecting the quality of the communication channel. 

Transmitting fresh information may be more beneficial when the probability of success is equal or even higher than retransmitting previously failed information. This is a well known trade-off between information freshness and reliability \cite{li2022age,soleymani2024transmit}. However, smart sensors have no full information regarding the quality of the channel that could allow them to decide whether to transmit new or retransmit an old information. Instead, they may only receive an acknowledgment feedback signal via \emph{(Hybrid) Automatic Repeat reQuest} protocols, to determine whether the transmitted information has been successfully received without errors. Using this protocol, the transmitter combines error correction techniques and retransmissions to ensure reliable delivery of data over noisy channels. Upon the reception of acknowledgment feedback signals (\texttt{ACK/NACK}), it updates a retransmission counter, and decides whether to discard a previously failed packet and transmit a fresh packet, or retransmit a previously failed one. However, relying only on the acknowledgement feedback to infer the quality of the channel, corresponds to an inexact and possibly delayed indicator of the \emph{channel state information} (CSI).

In \cite{you2011mean} the authors studied the stability condition for remote estimation over Markovian packet-dropping channels, focusing on transmitting the raw measurements instead of the pre-processed state estimate. Unlike \cite{you2011mean}, the results in \cite{qi2016optimal,chakravorty2019remote,huang2020real,liu2021remote}, considered the remote estimation problem over packet-dropping communication channels with Markov states, when transmitting the state estimate. In particular, these works proposed remote state estimation methods where the smart sensor can decide whether to retransmit a previously failed state estimate of higher \emph{age of information} (AoI) or to send a fresh estimate with lower probability of successful reception. 

The aforementioned works are developed under the assumption that both the smart sensor and the remote estimator have perfect knowledge of the channel state. However, in many cases, such an assumption might not hold due to rapidly changing communication channels, interference and noise, uncertainty, limited sensing range, etc. In this work, we focus on designing optimal transmission policies for remote state estimation over packet-dropping wireless channels with imperfect channel state information. The main contributions of this work are the following: (a) we formulate the problem of designing an optimal transmission policy over a wireless channel with partial CSI as a finite horizon partially observable Markov decision process (POMDP). Our approach, uses a combined state-space that incorporates both the AoI and the unknown channel state, (b) we identify an information state which satisfies the Markov property and serves as a sufficient statistic for the smart sensor's optimal decision making policy. Additionally, we derive the recursive relation for updating this information state; and (c) we solve the optimal decision problem  via dynamic programming, and numerically approximate the optimal transmission policy  via the point-based value iteration (PBVI) algorithm.


\section{System Model}
\label{sec:sys.model}

\subsection{Dynamical System and Smart Sensor}
\label{subsec:dyn.sys}
Consider a discrete-time LTI dynamical system 
\begin{subequations}\label{eq:lti_system}
\begin{align}
    x_{k}&=Ax_{k-1}+w_{k-1},\\
    y_k&=Cx_k+v_k,
\end{align}
\end{subequations}
where $k$ is the discrete time index, $x_k\in \mathbb{R}^{n}$ is the (unknown) state vector process, $y_k\in \mathbb{R}^{m}$ is the measurement vector process, $w_k\in \mathbb{R}^n$ is the process noise, $v_k\in\mathbb{R}^m$ is the measurement noise, and $A\in\mathbb{R}^{n\times n}$, $C\in\mathbb{R}^{m\times n}$ are known matrices. It is assumed that $x_0,w_0,w_1,\dots,v_0,v_1,\dots$ are mutually independent and Gaussian random variables $x_0\sim N(0,\Sigma_0)$, $w_k\sim N(0,R_w)$, $v_k\sim N(0,R_v)$ with $\Sigma_0 \succcurlyeq 0$, $R_w \succcurlyeq 0$, and $R_v \succ 0$. Here, $\cdot \succcurlyeq 0$ and $\cdot \succ 0$ denote  positive semidefiniteness and positive definiteness, respectively.
To estimate the unknown state $x_k$ from the noisy measurements $y_k$, the sensor implements a Kalman filter at each time step, leading to the following definitions.

Let us define the \emph{prior} and \emph{posterior} state estimates
\begin{align*}
    \hat{x}^s_{k|k-1}&:=\mathbb{E}[x_k|y_{1:k-1}],\\
    \hat{x}^s_{k|k}&:=\mathbb{E}[x_k|y_{1:k}],
\end{align*}
and their corresponding error covariances
\begin{align*}
    P^s_{k|k-1}&:=\mathbb{E}[(x_k-\hat{x}^s_{k|k-1})(x_k-\hat{x}^s_{k|k-1})^T|y_{1:k-1}],\\
    P^s_{k|k}&:=\mathbb{E}[(x_k-\hat{x}^s_{k|k})(x_k-\hat{x}^s_{k|k})^T|y_{1:k}],
\end{align*}
where $y_{1:k}$ denotes the sequence of noisy observations up to time $k$. The Kalman filter is given by the following prediction-correction equations 
\begin{subequations}\label{eq:kalman_filter}
\begin{align}
    \hat{x}^s_{k|k-1}&=A\hat{x}^s_{k-1|k-1},\\
    P^s_{k|k-1}&=AP^s_{k-1|k-1}A^T+R_w,\\
    K_k&=P^s_{k|k-1}C^T(CP^s_{k|k-1}C^T+R_v)^{-1},\\
    \hat{x}^s_{k|k}&=\hat{x}^s_{k|k-1}+K_k(y_k-C\hat{x}^s_{k|k-1}),\\
    P^s_{k|k}&=P^s_{k|k-1}-K_kCP^s_{k|k-1},
\end{align}
\end{subequations}
which are computed for each time step $k=1,2,\dots$. 


The smart sensor at each time $k$ must decide whether to transmit a new piece of information, \ie its current local state estimate $\hat{x}^{s}_{k|k}$, or retransmit old information, \ie an old local state estimate $\hat{x}^{s}_{k-\tau_k|k-\tau_k}$, as follows
\begin{align}\label{eq:packet_transmission}
    \hat{x}^s_k = \begin{cases}
        \hat{x}^s_{k|k}, &\text{if } \alpha_k=1,\\
        \hat{x}_{k-1}^{s}, &\text{if } \alpha_k=0,
    \end{cases}
\end{align}
where $\alpha_k\in\{$0,1$\}$ denotes the transmission action of the smart sensor at time $k$, and represents either the transmission of a new packet, \ie $\alpha_k=1$, or the retransmission of a previously failed packet, \ie $\alpha_k=0$. 

\begin{figure}[t]
    \centering
    \includegraphics[width=0.95\linewidth]{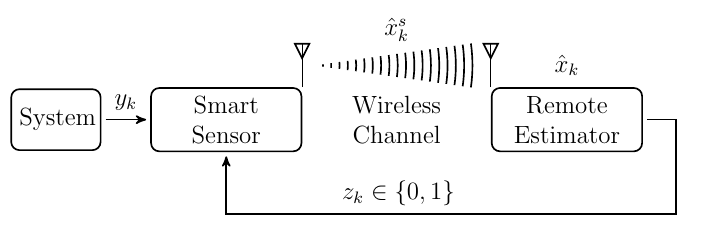}
    \caption{Remote state estimation scheme.}
    \label{fig:general-setup}
\end{figure}

\subsection{Wireless Channel and Remote Estimator}
\label{subsec:channel}
The wireless channel is considered as lossy and is modeled as a packet-erasure channel with ideal acknowledgment feedback. The quality of the channel can be modeled as a random process that varies with time in a correlated manner, although it is assumed constant within each time slot.
The state of the channel at time slot $k$ could be in one of the $n_c$ states that describe its quality, \ie $s^c_k \in \mathcal{S}^c \triangleq \{ 1, 2, \ldots, n_c \}$.
Thus, depending on the quality of the channel, a packet that is transmitted by the smart sensor to the remote estimator might arrive in error. In such a case, and with the remote estimator being incapable of correcting the error with the aid of the HARQ protocol, we say that a packet error occurred. 

The remote estimator is equipped with a HARQ protocol aiming at correcting possible errors that occurred due to channel fading. To indicate whether the packet sent by the smart sensor has been received by the remote estimator with or without errors, we introduce the  variable $z_k\in {\cal Z}\triangleq\{0,1\}$. That is, when a packet loss occurs at time $k$, the remote estimator sends a negative acknowledgment (\texttt{NACK}) back to the smart sensor, \ie $z_k=0$, while if the packet has been received without errors at time $k$, the remote estimator sends an acknowledgment (\texttt{ACK}), $z_k=1$. Here, it is important to note that the acknowledgement messages are assumed reliable. This is reasonable since acknowledgement messages are of one bit, and they are sent over narrowband and error-free feedback channels. Then, according to $z_k$, the remote state estimate $\hat{x}_k$ and its corresponding error covariance matrix $P_k$ are given by
\begin{equation}\label{eq.state.est}
    \hat{x}_k=\begin{cases}
			\hat{x}^s_{k}, & \text{if $z_{k}=1$,}\\
            A\hat{x}_{k-1}, & \text{otherwise,}
		 \end{cases}
\end{equation}
and
\begin{equation}\label{eq.error.cov}
    P_k=\begin{cases}
			P^s_{k|k}, & \text{if $z_k=1$,}\\
            AP_{k-1}A^T+R_w, & \text{otherwise}.
		 \end{cases}
\end{equation}

\subsection{Problem Setup}
The ultimate goal of the smart sensor is to minimize the mean squared error and energy consumption by optimally deciding between transmitting new information or re-transmitting  previously failed information. Given that the exact channel state is unknown, the smart sensor  relies its decision based  on its belief about the channel state. This belief is updated over time based on partial information inferred from an acknowledgment feedback signal received from the remote estimator, see Fig. \ref{fig:general-setup}.

Upon the reception of the feedback signal, the smart sensor is able to track the time elapsed since the last successful reception of information by the remote estimator. We denote this variable at time $k$ with $s^r_k \in \mathcal{S}^r \triangleq \{0,1,\ldots,n_r\}$, $n_r<\infty$, which is given by
\begin{align}\label{eq.AoI.state}
s^r_{k} = \begin{cases}
    1, & \text{ if } \alpha_{k-1} = 1 \land z_{k-1} = 0, \\
    s^r_{k-1}+1, & \text{ if } \alpha_{k-1} = 0 \land z_{k-1} = 0, \\
    0, & \text{ if } z_{k-1} = 1,
\end{cases}
\end{align}
for $k\geq0$ and with $s^r_0=0$. Clearly, when $s^r_k=0$ then only $\alpha_k=1$ can be taken. This is natural since $s^r_k=0$ represents the state where a fresh packet is to be transmitted. Moreover, when $s^r_{k-1}=n_r$, then we require $s^r_k=0$ and $\alpha_k=1$, since the number of retransmission attempts is reached.
To indicate that the AoI $s^r_{k}$ is a function of the previous state $s^r_{k-1}$, the  action $\alpha_{k-1}$ and the observation $z_{k-1}$, we write 
\begin{equation}
    s^r_{k}=L^r_{k-1}(s^r_{k-1},z_{k-1},\alpha_{k-1}).
\end{equation}

We assume that the packet error probability reduces exponentially with the number of retransmissions $r\in \mathcal{S}^r$, based on the model in \cite{ceran2019average}, that is,
\begin{equation}\label{HARQ.eqn}
    g(r,j) = q_{j} \lambda^{r},
\end{equation}
for $\lambda \in (0,1)$, and $q_j$ being the packet error probability of a freshly transmitted packet when in channel $j \in \mathcal{S}^c$. With smaller parameter $\lambda$, the probability of error decreases due to the ability of HARQ to decode the received packet correctly. With $\lambda=1$, the HARQ protocol reduces to the classical ARQ protocol without combining, where the probability of packet error is not reduced, but it is instead constant for all subsequent retransmissions. 

The AoI variable $s^r_{k}$, allows the smart sensor to determine the current state estimate of the remote estimator at time $k$, where from \eqref{eq.state.est} and \eqref{eq.error.cov}, we get
\begin{align}
    \hat{x}_{k} &= \begin{cases}
        \hat{x}^s_k, & \text{ if } s^r_k=0,\\
        A^{s^r_k} \hat{x}_{k-s^r_k}, & \text{ otherwise,}
    \end{cases}   
\end{align}
and
\begin{equation}\label{cov.eq.AoI}
    P_{k} = \begin{cases}\!
        P^s_{k|k}, &\text{if  $s^r_k{=}0$,}\\
        A^{s^r_k}P^s_{k-s_k^r}(A^{s^r_k})^T{+}\sum\limits_{j=0}^{{s^r_k}-1}A^j R_w (A^j)^T ,& \text{otherwise.}
    \end{cases}   
\end{equation}


Note that, although HARQ protocols improve transmission reliability by reducing the packet error probability with each retransmission, they also introduce challenges, such as increased latency and energy consumption. In this work, HARQ will be explicitly incorporated into the decision-making framework by optimizing transmission decisions that balance estimation accuracy, AoI, and energy consumption.

\section{Problem Formulation}\label{sec.problem.form}
\subsection{Finite horizon POMDP}
A discrete-time POMDP on a finite horizon is a collection
\begin{equation}\label{POMDP}
    ({\cal S}, {\cal A}, {\cal Z},  T(\alpha),  O(\alpha), c_k(\alpha), c_N)
\end{equation}
consisting of the following elements:
\begin{enumerate}
    \item[a)] The state space ${\cal S}:={\cal S}^c\times {\cal S}^r$ where 
 ${\cal S}^c:=\{1,2,\dots,n_c\}$ denotes the set of possible channel states, and
 ${\cal S}^r:=\{0,1,\dots,n_r\}$ denotes the set of possible AoI values. Each  $s:=(s^c,s^r)\in S$ represents the system being in channel state $s^c\in {\cal S}^c$ and having an AoI $s^r\in {\cal S}^r$.
    \item[b)] The action space ${\cal A}:=\{0,1\}$. Each $\alpha\in {\cal A}$ represents the action that the smart sensor can take. 
    \item[c)] The observation space ${\cal Z}:=\{0,1\}$. Each $z\in {\cal Z}$ provides information on whether or not the package has been successfully received by the remote estimator.
    \item[d)] The transition probability matrix $T(\alpha)$ defined by
    \begin{align}\label{POMDP.trans}
        T(\alpha)&:=\operatorname{Pr}(s_{k+1}|s_k,\alpha_k)=\operatorname{Pr}(s^c_{k+1},s^r_{k+1}|s^c_{k},s^r_{k},\alpha_k)\nonumber\\
        &=\operatorname{Pr}(s^c_{k+1}|s^c_{k},s^r_{k+1},\alpha_k)\operatorname{Pr}(s^r_{k+1}|s^c_{k},s^r_{k},\alpha_k)\nonumber\\
        &\overset{(a)}=\operatorname{Pr}(s^c_{k+1}|s^c_k)\operatorname{Pr}(s^r_{k+1}|s^c_k,s^r_k,\alpha_k),
    \end{align}
    where (a) stems from the fact that the channel state $s_{k+1}$ depends only on the current state $s_{k}$.
    Let $T^c=[\tau^c_{i,j}]$, $i,j\in {\cal S}^c$, be the channel state transition probability matrix with 
    \begin{equation}
        \tau^c_{i,j}=\operatorname{Pr}(s^c_{k+1}=j|s^c_{k}=i),
    \end{equation}
    denoting the $(i,j)$--th element of $T^c$.
    Let $T^r(\alpha)=[\tau^r_{i,\ell,q}(\alpha)]$, $i\in {\cal S}^c$, $\ell,q\in {\cal S}^r$, $\alpha\in {\cal A}$, be the AoI state transition probability matrix with
    \begin{equation}
        \tau^r_{i,\ell,q}(\alpha){=} \operatorname{Pr}(s^r_{k+1}=q|s^c_{k}=i,s^r_{k}=\ell,\alpha_k{=}\alpha),
    \end{equation}
    denoting the  $(i,\ell,q)$--th element of  $T^r(\alpha)$.
    \item[e)] The observation probability matrix $O(\alpha)=[o_{i,j,m}(\alpha)]$, $i\in {\cal S}^c$, $j\in {\cal S}^r$, $m\in {\cal Z}$, $\alpha\in {\cal A}$, defined by
        \begin{equation}\label{POMDP.op}
        O(\alpha):=\operatorname{Pr}(z_k|s_k,\alpha_k)=\operatorname{Pr}(z_k|s^c_k,s^r_k,\alpha_k),
    \end{equation}
    with 
    \begin{equation}
        o_{i,\ell,m}(\alpha)=\operatorname{Pr}(z_k=m|s^c_k=i,s^r_k=\ell,\alpha_k{=}\alpha),
    \end{equation}
    denoting the $(i,\ell,m)$--th element of $O(\alpha)$.
    \item[f)] The cost function $c_k(s,\alpha)$ incurred at time $k$ when the state is $s\in {\cal S}$ and action $\alpha\in {\cal A}$ is applied.
    \item[g)] The terminal cost function $c_N(s)$ incurred at  time $k=N$ when the state is $s\in {\cal S}$.
\end{enumerate}

The POMDP defined by \eqref{POMDP} represents a partially observable model in which the augmented system state $s_k\in {\cal S}$ consists of both known and unknown components. In particular, the  AoI $s^r\in {\cal S}^r$ is the known part of the system state while the channel state $s^c\in {\cal S}^c$ is the unknown (unobserved) part of the system state. Hence, the smart sensor must infer the unknown part of the system state  through some indirect observations $z_k\in {\cal Z}$ at each time $k$. 

To specify the information available to the smart sensor, let us define the set 
\begin{equation*}
    I_k:=\{s^r_0,z_0,\dots,z_{k-1},\alpha_0,\dots,\alpha_{k-1}\},\ k=1,2,\dots,N.
\end{equation*}
Notice from \eqref{eq.AoI.state} that given the information history  $I_k$, then the AoI state $s^r_k\in {\cal S}^r$ can be calculated for all $k$. It follows that, for each $k$, $I_k$ conveys the same amount of information as $\bar{I}_k:=\{I_k,s^r_1,\dots,s^r_k\}$.
Let us consider the set of all policies compatible with $I_k$. For any $g=\{g_0,g_1,\dots,g_{N-1}\}$, $g_k:I_k\longmapsto {\cal A}$ is feasible if
\begin{equation}
    \alpha_k=g_k(I_k)\in {\cal A},\quad \forall I_k, \ k=1,2,\dots,N-1.
\end{equation}
Let $G$ denote the set of all feasible policies. 

The next step in solving a POMDP is to define an information state, also referred to as the belief state. The information state at time $k$, denoted as $\pi_k$, is the posterior distribution over the unknown component of the system state conditioned on the available information,  that is,
\begin{equation}\label{belief.state}
\pi_k(j|I_k):=\operatorname{Pr}(s^c_k=j|I_k),\quad \forall j\in {\cal S}^c.
\end{equation}

Using Bayes' rule, the recursive relation of the information state from time $k$ to time $k+1$ is given by 
\begin{equation}\label{belief.recursion}
    \pi_{k+1}(s^c_{k+1}=j|I_{k+1})
    =\frac{\displaystyle \sum_{i\in {\cal S}^c}\tau^{r,z}_{i,\ell,q,m}(\alpha)\tau^c_{i,j}\pi_k(i|I_k)}{\displaystyle\sum_{i\in {\cal S}^c}\tau^{r,z}_{i,\ell,q,m}(\alpha)\pi_k(i|I_k)},
\end{equation}
with
\begin{align}\label{joint.distrib}
\tau^{r,z}_{i,\ell,q,m}(\alpha)&=\operatorname{Pr}(s^r_{k+1}=q|s^c_k=i,s^r_k=\ell,z_k=m,\alpha_k=\alpha)\nonumber\\
&\ \times \operatorname{Pr}(z_k=m|s^c_k=i,s^r_k=\ell,\alpha_k=\alpha).
\end{align}
To indicate that the information state is a function of the prior distribution $\pi_k$, the  action $\alpha_k$ and the observation $z_{k}$, we write 
\begin{equation}\label{eq.post}
    \pi_{k+1}(I_{k+1})=L_k(\pi_k(I_k),z_{k},\alpha_k),
\end{equation}
where $\pi_k(I_k)=(\pi_k(1|I_k),\dots,\pi_k(n_c|I_k))$ is an $n_c$--dimensional row vector. The proof of \eqref{belief.recursion} is omitted due to space limitations.

\begin{remark} 
        By utilizing the fact that (a) the information state $\pi_k$ provides a sufficient statistic for the information available to the smart sensor, that is,
        $\alpha_k=g_k(I_k)=g_k(\pi_k)\in {\cal A}$, and (b) for any fixed sequence of actions $\{\alpha_0,\dots,\alpha_k\}$, the sequence  $\{\pi_0,\dots,\pi_k\}$ is a Markov process, then the POMDP \eqref{POMDP} is equivalent to the corresponding fully observed Markov decision process \cite{varayia86}.
\end{remark}

\subsection{Performance Criterion}\label{subsec.performance.crit}
Each policy $g\in G$ incurs a series of costs 
\begin{equation}
c_k(s_k,\alpha_k)=\mbox{Tr}(P_k)+\epsilon_j(\alpha), 
\end{equation}
where $P_k$ is specified by \eqref{cov.eq.AoI}, and $\epsilon_j(\alpha)$ denotes the energy consumption modeled as a deterministic function of the channel state  $s^c_k=j\in S^c$ and action $\alpha_k=\alpha\in {\cal A}$.
    In HARQ protocols, retransmissions often involve transmitting a packet at a higher power level due to additional encoding to improve the chances of decoding the message at the receiver, especially when previous attempts have failed due to channel conditions. This attempt raises the energy cost $\epsilon_j$ for retransmitted packets compared to a simple new packet transmission.
Define the $N$--stage expected cost  by 
\begin{equation}\label{finite.cost}
    J(g):=\mathbb{E}^g\Big[\sum_{k=0}^{N-1}c_k(s_k,\alpha_k)+c_N(s_N)\Big],
\end{equation}
where $\mathbb{E}^g[\cdot]$ indicates the dependence of the expectation operation on the policy $g\in G$, and $c_N$ denotes the terminal cost, a given deterministic function $c_N(s_N):=\mbox{Tr}(P_N)$.

The optimal decision problem is that of selecting the admissible policy $g\in G$ such that the performance criterion \eqref{finite.cost} is minimized. In particular, the optimal decision problem is to choose $g^*\in G$ such that
\begin{equation}\label{ACOP}
J(g^*)=\inf_{g\in G}J(g)=J^*.
\end{equation}
A policy $g^*\in G$ that satisfies \eqref{ACOP} is called an optimal policy, and the corresponding $J^*$ is called the minimum cost.

\section{Problem Solution}\label{sec.Sol}
In this section, we provide via dynamic programming the solution of the optimal decision problem \eqref{ACOP}.

Let $V_k(\alpha^g_{[k,N-1]},I_k)$ denote the value function on the time horizon $k,k+1,\dots,N$ given an optimal policy $g^*_t$, $t=1,2,\dots,k-1$ defined by
\begin{equation}
    V_k(\alpha^g_{[k,N-1]},I_k) = \mathbb{E}\Big[\sum_{t=k}^{N-1} c_t(s_t^g,\alpha_t^g)+c_N(s_N^g)|I_k\Big],
\end{equation}
where $\alpha^g_{[k,N-1]}$ denotes the restriction of policies in $[k,N-1]$. 
Let us also define the probability row vector $\pi=(\pi(1),\dots, \pi(n_c))\in \Pi_1({\cal S}^c)$ where
\begin{equation*}
    \Pi_1({\cal S}^c):=\{\pi\in \mathbb{R}^{n_c}:\sum_{i\in {\cal S}^c}\pi(i)=1, \pi(i)\geq 0, \forall i\in {\cal S}^c\}.
\end{equation*}


The main theorem of this paper follows.
\begin{theorem}\label{main.theorem}
    Define recursively the functions $V_k(\pi,s^r)$, $0\leq k\leq N$, $\pi\in \Pi_1({\cal S}^c)$, $s^r\in {\cal S}^r$ by
    \begin{subequations}\label{DP.equation.main}
    \begin{align}\label{DP.equation.a}
        &V_N(\pi,s^r)=\mathbb{E}[c_N(s^c_N,s^r)|\pi_N(I_N)=\pi,s^r_N=s^r],\\
        &V_k(\pi,s^r)=\min_{\alpha_k\in {\cal A}}\mathbb{E}[c_k(s^c_k,s^r,\alpha_k) \label{DP.equation.b}\\
        & {+} V_{k+1}(L_k(\pi,z_{k},\alpha_k), L^r_k(s^r,z_k,\alpha_k))|\pi_k(I_k){=}\pi,s^r_k{=}s^r].\nonumber
    \end{align}
    \end{subequations}
    \begin{enumerate}
        \item Let $g\in G$. Then,
        \begin{equation}\label{thm.part.a}
            V_k(\pi_k(I_k),s^r_k)\leq V_k(\alpha^g_{[k,N-1]},I_k).
        \end{equation}
        \item Let $g\in G$ be such that for all $\pi\in \Pi_1({\cal S}^c)$, $g_k(\pi)$ achieves the infimum in \eqref{DP.equation.b}. Then, $g$ is optimal and
        \begin{equation}\label{thm.part.b}
            V_k(\pi_k(I_k),s^r_k)=V_k(\alpha^g_{[k,N-1]},I_k).
        \end{equation}
    \end{enumerate}
\end{theorem}

\begin{proof}
    By following similar steps as  in  \cite{varayia86}. 
\end{proof}

The expectation in the right side of \eqref{DP.equation.a} is given by
\begin{align*}
&\sum_{j\in {\cal S}^c}c_N(j,s^r)\operatorname{Pr}(j|\pi_N(I_N)=\pi,s^r_N=s^r)\\
&=\sum_{j\in {\cal S}^c}c_N(j,s^r)\operatorname{Pr}(j|\pi_N(I_N)=\pi)=\sum_{j\in {\cal S}^c}c_N(j,s^r)\pi(j) ,
\end{align*}
where the first equality follows since $\pi_N(I_N)=\pi$ summarizes all the information of $I_N$, including the information of $s^r_N=s^r$.
The expectation in the right side of \eqref{DP.equation.b} is 
\begin{align*}
 &   \sum_{j\in {\cal S}^c}c_k(j,s^r,\alpha)\pi(j)\ \mathclap{+}\sum_{m\in {\cal Z}
 }V_{k+1}(L_k(\pi,m,\alpha),L^r_k(s^r,m,\alpha))\\
&\quad\times\operatorname{Pr}(m|\pi_k(I_k)=\pi,s_k^r=s^r,\alpha_k=\alpha),\nonumber
\end{align*}
where the conditional probability of $z_k=m$ given $\pi_k(I_k)=\pi$, $s^r_k=s^r$, and $\alpha_k=\alpha$, is 
\begin{align*}
    &\operatorname{Pr}(z_{k}{=}m|\pi_k(I_k){=}\pi,s^r_k{=}s^r,\alpha_k{=}\alpha)
    =\sum_{i\in {\cal S}^c}  o_{i,s^r,m}(\alpha)\pi(i).
\end{align*}

\begin{figure*}[t]
    \centering
    \begin{subfigure}[t]{0.475\linewidth}
         \centering
         \includegraphics[width=\linewidth]{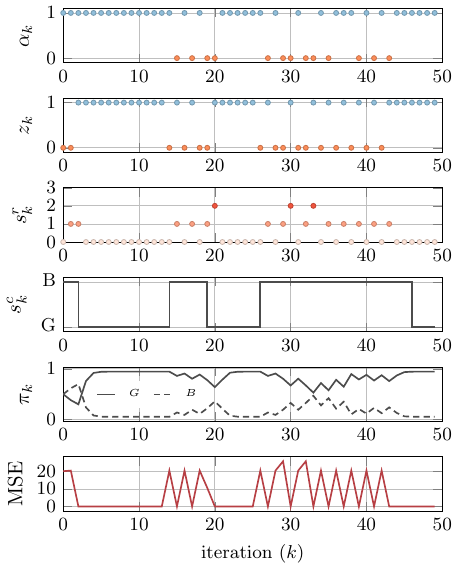}
         \caption{HARQ under $T^c_1$ with $\lambda=0.5$.}
         \label{fig:harq_pbvi}
    \end{subfigure}
    \hfill
    \begin{subfigure}[t]{0.475\linewidth}
         \centering
         \includegraphics[width=\linewidth]{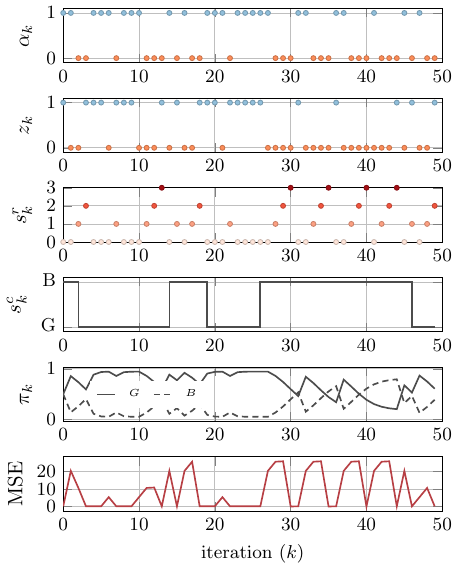}
         \caption{ARQ under $T^c_1$ with $\lambda=1$.}
         \label{fig:arq_pbvi}
    \end{subfigure}
    \caption{Evolution of the POMDP states and the MSE per iteration $k$.}
    \label{fig:sensor}
\end{figure*}

The PBVI algorithm \cite{smallwood1973optimal} is employed to numerically  approximate the optimal policy. The PBVI algorithm, using a  finite set of representative belief points $\mathcal{B}_\pi=\{\pi^{(1)},\dots, \pi^{(M)}\}$,  $\pi^{(\cdot)}\in \Pi_1({\cal S}^c)$, approximates the value function  and iteratively solves the dynamic programming \eqref{DP.equation.main} for each AoI state. 
This recursive process allows the algorithm to compute an optimal transmission policy while considering both the AoI and channel state uncertainty, resulting in a policy that balances between information freshness and reliability of the channel state estimation.

\section{Numerical Experiments}\label{sec:numerical_experiments}
Consider a discrete-time LTI dynamical system with
$$
    A = \begin{bmatrix}
        0.9974 & 0.0539\\
        -0.1078 & 0.1591
    \end{bmatrix},\quad 
    C= \begin{bmatrix}
        1 & 0
    \end{bmatrix}.$$
The smart sensor employs the Kalman filter defined in \eqref{eq:kalman_filter} with $R_w=\text{diag}(0.25,0.25)$, and $R_v = 0.05$.
 We model the channel, over which the packets are transmitted, as a two-state time homogeneous Markov chain. The channel can be either in a good ($G$) or bad ($B$) state. We consider two different transition probability  matrices 
$$
T^c_1 = \begin{bNiceArray}{*2c}[
first-row,code-for-first-row=\scriptstyle \color{black!50},
last-col,code-for-last-col=\scriptstyle \color{black!50}
]
    G & B & \\
    0.95 & 0.05 & G\\
    0.1 & 0.9 & B
\end{bNiceArray}\;\;,\quad
T^c_2 = \begin{bNiceArray}{*2c}[
first-row,code-for-first-row=\scriptstyle \color{black!50},
last-col,code-for-last-col=\scriptstyle \color{black!50}
]
    G & B & \\
    0.5 & 0.5 & G\\
    0.5 & 0.5 & B
\end{bNiceArray}\;\;,
$$
with invariant probability distributions $p=[0.667\ 0.333]$ and $p=[0.5\ 0.5]$, respectively.
The probability of error of a packet transmission depends on the quality (state) of the channel. Hence, we assume that the probability of error of a fresh (new) transmission when in state $G$ is $q_G=0.2$, while in state $B$ is $q_B=0.8$. In this example, we allow the smart sensor to possibly retransmit the same packet for $n_r=3$ times. Hence, the probability of error of a packet (re)-transmitted for $r\in\mathcal{S}^r$ times is given by \eqref{HARQ.eqn}, with $\lambda \in (0,1]$.

For the recursive update of the belief state $\pi_k$, \eqref{belief.recursion} is used, with its right side specified by
\begin{align*}
\tau^r_{i,\ell,q}(0) &= \begin{cases}
    p_i, & \text{ if } m = 0 \land q=0 \land \ell=n_r,\\
    p_i, & \text{ if } m = 0 \land q=\ell+1 \land 0< \ell<n_r,\\
    p_i, & \text{ if } m = 1 \land q=0 \land \ell>0, \\
    0, & \text{ otherwise, } 
\end{cases}\\
\tau^r_{i,\ell,q}(1) &= \begin{cases}
    p_i, & \text{ if } m = 0 \land q=1,\\
    p_i, & \text{ if } m = 1 \land q=0,\\
    0, & \text{ otherwise, } 
\end{cases}\\
o_{i,\ell,0}(\alpha_k) &= \begin{cases}
    0, & \text{ if } \alpha_k = 0 \land \ell=0,\\
    q_i\lambda^\ell, & \text{ if } \alpha_k = 0 \land \ell>0,\\
    q_i, & \text{ if } \alpha_k = 1,
\end{cases}\\
o_{i,\ell,1}(\alpha_k)&=1-o_{i,\ell,0}(\alpha_k),
\end{align*}
where $\ell,q=0,1,\dots,n_r$, $i\in \{G,B\}$.



Fig. \ref{fig:sensor}\subref{fig:harq_pbvi} and Fig. \ref{fig:arq_pbvi} depict the execution of the  aforementioned setup using the same realization of $T^c_1$, for  $\lambda=0.5$ and $\lambda=1$, respectively. In both figures, we record the variables held by the smart sensor, and the mean squared error (MSE) at the remote estimator. 
In the following, we compare the results obtained under HARQ and ARQ protocols:
\begin{list4}
\item Under both protocols, the sensor must balance the trade-offs between improving the MSE (via lower AoI) and saving energy, particularly when retransmissions are required due to the bad channel state $s^c_k=B$.
\item Comparing the two protocols, we  observe that the main difference lies in how AoI evolves. 
Combined with the sensor's actions $\alpha_k$, HARQ reduces the number of transmissions, particularly under poor channel conditions.
\item The belief state $\pi_k$ is updated at each time $k$  based on  $z_k$ and  $\alpha_k$. When 
{\ack}  is received, $\pi_k$ remains near $s^c_k=G$, indicating high confidence in the current channel state. Upon 
\nack, $\pi_k$ shifts, indicating increased uncertainty about the channel state. This shift occurs because a {\nack} suggests a potential channel state, prompting $\pi_k$ to increase the probability of $s^c_k=B$.
\end{list4}

Finally, Fig. \ref{fig:MSE_function_lambda} depicts MSE averaged over $100$ realizations of the Markov chains $T^c_1$ and $T^c_2$, with each result corresponding to a specific value of $\lambda$. The Markov chain $T^c_2$  represents the worst-case scenario since it implies complete uncertainty about the channel state transitions. As a result, the MSE  is higher than the one corresponding to $T^c_1$.


\section{Conclusions}\label{sec.concl}
We  design an optimal transmission policy for remote state estimation over packet-dropping wireless channels with imperfect channel state information. 
We formulate  this problem as a  finite horizon POMDP with an augmented state-space that incorporates both AoI and  unknown channel state. An information state  serves as a sufficient statistic for decision-making, and the optimal policy is obtained via dynamic programming and solved numerically using the PBVI algorithm.

\begin{figure}[t]
    \centering
    \includegraphics[width=0.8\linewidth]{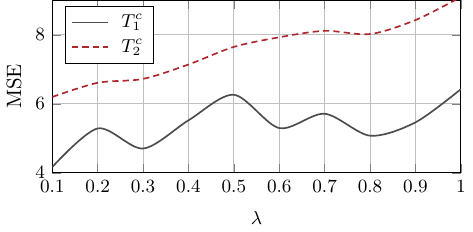}
    \caption{MSE for different realizations of  $T^c_1$ and $T^c_2$\vspace{-0.23cm}.}
    \vspace{-0.23cm}
    \label{fig:MSE_function_lambda}
\end{figure}

\bibliographystyle{ieeetr}
\bibliography{references}

\end{document}